\documentclass[11pt]{article}
\usepackage[margin=1in]{geometry}
\usepackage[utf8]{inputenc}
\usepackage{mathrsfs}
\usepackage{amssymb,latexsym,amsmath}
\usepackage{amsthm}
\usepackage{color}
\usepackage{xspace}
\usepackage{tcolorbox}
\usepackage{hyperref}
\usepackage{thmtools}
\usepackage{cleveref}
\usepackage{wrapfig}
\usepackage{thm-restate}
\usepackage{subcaption}
\usepackage{tikz}

\newtheorem{theorem}{Theorem}
\newtheorem{lemma}{Lemma}[section]
\newtheorem{definition}[lemma]{Definition}

\newtheorem{observation}[lemma]{Observation}

\newtheorem{claim}[lemma]{Claim}

\newtheorem{problem}[lemma]{Problem}

\crefname{lemma}{Lemma}{Lemmas}
\crefname{appendix}{Appendix}{Appendices}
\crefname{proposition}{Proposition}{Propositions}
\crefname{observation}{Observation}{Observations}
\crefname{claim}{Claim}{Claims}
\crefname{figure}{Figure}{Figures}

\title{Vantage Point Selection Algorithms for Bottleneck Capacity Estimation}
\author{Vikrant Ashvinkumar\thanks{Rutgers University. Email: {\tt va264@rutgers.edu}} \and Rezaul Chowdhury\thanks{Stony Brook University. Email: {\tt rezaul@cs.stonybrook.edu} Partially supported by NSF (CCF-2318633).} \and Jie Gao\thanks{Rutgers University.  Email: {\tt jg1555@rutgers.edu} Ashvinkumar and Gao are supported by NSF through IIS-2229876, DMS-2220271, DMS-2311064, CCF-2208663, CCF-2118953.} \and Mayank Goswami\thanks{Queens College, City University of New York. Email: {\tt mayank.goswami@qc.cuny.edu} Supported by NSF grant CCF-2503086.} \and Joseph S. B. Mitchell\thanks{Stony Brook University.  Email: {\tt joseph.mitchell@stonybrook.edu} Partially supported by NSF (CCF-2007275).} \and Valentin Polishchuk\thanks{Linkoping University. Email: {\tt valentin.polishchuk@gmail.com}}}
\date{}

\newcommand{\Ot}{\ensuremath{\widetilde{O}}}

\newcommand{\bracket}[1]{\left[#1\right]}
\newcommand{\paren}[1]{\ensuremath{\left(#1\right)}\xspace}

\newcommand{\card}[1]{\left\vert{#1}\right\vert}

\newcommand{\prob}[1]{\Pr\bracket{#1}}

\newcommand{\Exp}{\mathbb{E}}
\newcommand{\expect}[1]{\Exp\bracket{#1}}

\newcommand{\set}[1]{\ensuremath{\left\{ #1 \right\}}}

\newcommand{\VC}{\ensuremath{\mathsf{VertexCover}\xspace}}
\newcommand{\BD}{\ensuremath{\mathsf{BottleneckDiscovery}\xspace}}

\newcommand{\GoodL}{\ensuremath{\mathsf{CountGoodLabellings}}\xspace}
\newcommand{\BWL}{\ensuremath{\mathsf{CountBlackWhiteColorings}}\xspace}
\newcommand{\GoodLN}{\textsc{CountGoodLabellings}\xspace}
\newcommand{\BWLN}{\textsc{CountBlackWhiteColorings}\xspace}

\def\eps{\varepsilon}
\newcommand{\E}{\mathbb{E}}

\newenvironment{algorithm}
{\begin{tabular}{|l|}\hline\begin{minipage}{1in}\normalfont\begin{tabbing}
        \quad\=\qquad\=\qquad\=\qquad\=\qquad\=\qquad\=\qquad\=\qquad\=\qquad\=\qquad\=\qquad\=\qquad\=\qquad\=\kill}
{\end{tabbing}\end{minipage}\\\hline\end{tabular}}
\newenvironment{algo}
{\begin{center}\small\begin{algorithm}}
{\end{algorithm}\end{center}}

\begin{document}
\maketitle
\begin{abstract}
Motivated by the problem of estimating bottleneck capacities on the Internet, we formulate and study the problem of vantage point selection. We are given a graph $G=(V, E)$ whose edges $E$ have unknown capacity values that are to be discovered. Probes from a vantage point, i.e, a vertex $v \in V$, along shortest paths from $v$ to all other vertices, reveal bottleneck edge capacities along each path. Our goal is to select $k$ vantage points from $V$ that reveal the maximum number of bottleneck edge capacities. 

We consider both a non-adaptive setting where all $k$ vantage points are selected before any bottleneck capacity is revealed, and an adaptive setting where each vantage point selection instantly reveals bottleneck capacities along all shortest paths starting from that point.
In the non-adaptive setting, by considering a relaxed model where edge capacities are drawn from a random permutation (which still leaves the problem of maximizing the expected number of revealed edges NP-hard), we are able to give a $1-1/e$ approximate algorithm.
In the adaptive setting we work with the least permissive model where edge capacities are arbitrarily fixed but unknown.
We compare with the best solution for the particular input instance (i.e. by enumerating all choices of $k$ tuples), and provide both lower bounds on instance optimal approximation algorithms and upper bounds for trees and planar graphs.
\end{abstract}

\section{Introduction}

Network tomography~\cite{Vardi1996-cy, Castro2004-ll,Coates2002-xq} concerns the problem of inferring the internal topology and parameters of a network based on end-to-end measurements. For example, collecting IP, router, and provider-level network topologies has been an active research topic for more than 20 years. Long-term continuous measurement efforts such as CAIDA's Ark infrastructure~\cite{Archipelago} provide important information useful to many longitudinal analyses and network events of interest. One of the basic measurement tools is \textsc{TraceRoute}, which obtains the sequence of router interface IP addresses along the forward path to a destination by sending probe packets with varying time to live (TTL) values and examining the Internet Control Message Protocol (ICMP) responses. By using the transmission timestamp of each probe, \textsc{TraceRoute} can report the round trip time from the source to \emph{each} node on the forward path. Although \textsc{TraceRoute} was initially designed for network administrators to make diagnoses on a small scale, recent efforts applied \textsc{TraceRoute} for Internet scale topology probing~\cite{Huang2020-uy, Beverly2016-jn}. 

Beyond network topology, estimating network bottleneck capacity is a classical research topic~\cite{Allman1999-ex, Guerrero2010-bx, Banerjee2000-cn, Harfoush2003-mw}. The capacity of a link (edge) is the highest bit rate possible to transmit over the edge. The capacity of a path is determined by the \emph{bottleneck} edge, the one on the path with the lowest capacity~\cite{Prasad2003-sv}. Estimating the bottleneck or available capacity between a pair of vertices is useful for many applications, such as routing management, intrusion detection and improving performance of transport protocols. In particular, knowledge of network bandwidth can be instrumental to client-side optimization in real-time throughput sensitive applications, such as video conferencing. 
A list of existing measurement tools for available bandwidth estimation is summarized in~\cite{Salcedo2018-kx}.
Existing techniques mostly focus on end-to-end capacity, using tools such as variable packet size, packet pair/train dispersion, or periodic streams. Common to these different implementations is the central idea of examining the \emph{jumps} in round trip time (RTT) from the source to each hop of a path. Since the bottleneck edge determines the available capacity of the entire path, the bottleneck capacity can be discovered by examining packet latency in probes.

The common practice in large scale measurements of the Internet issues probes from a narrow range of vantage points (VPs), typically cloud systems or research universities, due to cost, convenience, accessibility and scalability~\cite{Jueckstock2019-ym}. For example, RIPE Atlas~\cite{RIPE-Atlas}, a platform that supports user specified measurement requests, has over $9600$ active measurement vantage points. But a user is limited to the total number of vantage points used in a measurement. Therefore, a user would naturally seek to maximize the amount of information one could obtain by optimizing the choice of vantage points~\cite{Holterbach2017-ec}. 

In this paper we formulate mathematically the problem of choosing vantage points for network probing. Our goal is to use probing tools for network-wide capacity discovery. 
We consider a network with publicly known topology, where each edge/link has an unknown capacity. We would like to select (possibly from a given set of vertices) a set of $k$ vantage points from which probing messages are sent out to other vertices in the network.
We use a model of probing in which probing from a source $s$ along a path to a destination $t$ reveals the bottleneck capacity of the sub-path from $s$ to $v$, for every intermediate vertex $v$ on the path. In this model, if the bottleneck capacity drops from $C_1$ to $C_2$, in comparing paths $P(s, v)$ and $P(s, v')$ with $v'$ being the immediate downstream vertex after $v$, then the edge $(v, v')$ has capacity $C_2$.
Our objective is to reveal the link capacity for as many edges in the network as possible. 

In this paper we assume that the network topology is known, e.g., learned using the relatively mature methods for topology discovery. We also work mainly with an assumption that all paths along the probings are unique shortest paths. Thus the set of shortest paths from one vertex to all other vertices in the network forms a shortest path tree. To examine the connection to reality, we first remark that modern Internet uses multi-path routing, where traffic to a destination can be spread on multiple paths for throughput and redundancy~\cite{Augustin2007-pf, Vermeulen_undated-sj}. On the other hand, the assumption that routes originating from a vantage point to the Internet form a tree-like structure is heavily exploited to reduce probing redundancy in modern \textsc{TraceRoute} tools, e.g., Doubletree~\cite{Donnet2005-pe} and its variants~\cite{Donnet2007-dq, Li2008-eq}. Therefore it is reasonable to assume that using existing tools the routes probed starting from a vantage point indeed form a tree. 

There has been a number of empirical work understanding  vantage points on the Internet in terms of their characteristics and influence on data collected for network probing~\cite{Burger2014-hb, Bottger2017-oc, Lichtblau2021-tr, Shavitt2011-wk}. None of the prior work is directly related to our problem formulation. 

The problem of vantage point selection is not limited to the Internet domain. There are a number of other network scenarios where estimating bottleneck on the network is important. This problem is natural for transportation network for estimating traffic bottleneck although probing on the transportation network may have to rely on opportunistic inputs. 
In the traffic engineering literature, the traffic sensor location problem (TSLP)~\cite{OWAIS2022118134} is to determine how many sensors are required and where should they be deployed in order to best understand the traffic bottlenecks in road and transportation networks. 
For a blood vessel network, a medical procedure known as angiography considers probing of the blood vessel network of a patient by a dye injected into the bloodstream through a catheter, in order to determine the blocked vessel (the bottleneck). On an abstract level, determining the site for the injection(s) is similar to the bottleneck discovery problem in this paper. 

\subsection{Problem Definition}\label{sec:overview}

Here we formulate the problem of vantage point selection for bottleneck capacity discovery.
We are given an undirected graph $G=(V, E)$ with each edge  associated with an \emph{unknown} positive capacity $c(\cdot)$.
The goal is to discover these unknown capacities by using bottleneck queries from vantage points $u$, which reveals bottlenecks on the shortest path between $u$ and $v$ for all $v \in V$.
The shortest path may be defined by the path of minimum number of hops or by another weight metric which is known (the only unknown is the capacities).
Unless mentioned otherwise, we will assume (perhaps with slight perturbation of the edge weights) that all shortest paths are unique.
The problem is then to select a set of $k$ vertices $S$ as vantage points such that when we issue queries from each vantage point, we can reveal a maximum number of the unknown edge capacity values.

We first discuss what edges are revealed from selecting $S$ as a set of vantage points.

\begin{definition}[Bottleneck Edges]
\label{def:bn}
    Denote by $P(s,t)$ the shortest path from $s$ to $t$, and let $e_1, e_2, \ldots e_\ell$ be the edges along this path.
    We say that, for $i \in [\ell]$, $e_i$ is a \emph{bottleneck edge} along $P(s,t)$ if $c(e_i) < c(e_j)$ for all $j \in [\ell] \setminus \set{i}$, and $c(e_i)$ is a \emph{bottleneck capacity}.
    The bottleneck edges from a set $S$ to $V$ are edges that are bottleneck along $P(s,t)$ for some $s \in S$ and some $t \in V$.
\end{definition}

Note that while we state the definition of a bottleneck edge with respect to a path $P(s,t)$, in accord with the standard definition of a bottleneck edge,
we extend the definition to a set $S$ to match our query model:
for every $s \in S$, a single query comprises of $n-1$ subqueries finding the bottleneck edge in $P(s,t)$ for all $t \in V \setminus \set{s}$.
For example, on the path graph $v_0, v_1, \ldots, v_n$ where edges have monotone increasing capacities, querying $v_0$ results in the revelation of all capacities whereas querying only reveals the capacities of the last edge.

With this background, we can loosely define the problem of bottleneck discovery.
\begin{problem}[$\BD$]
    Given a non-negatively weighted (weights $w$) undirected graph $G=(V,E,w)$, and a positive integer $k$, find a subset of vertices $S \subseteq V$ with $|S| = k$ such that the number of revealed bottleneck edges is maximized.
\end{problem}

\smallskip\noindent\textbf{Non-Adaptive Setting. } In this setting, we work with the assumption that the ordering of the edge capacities is uniformly randomly chosen from a permutation of $1, 2, \ldots, m=|E|$. Thus, each edge has equal chance of being the one with the smallest capacity among all edges in the network.  
Here we consider algorithms in a non-adaptive manner, where the $k$ vantage points are selected all at once and probes are issued afterwards. We seek to maximize the \emph{expected} number of edges revealed.

\smallskip\noindent\textbf{Adaptive Setting. }
In the adaptive setting, we work in the worst-case model where we assume the edge capacities to be arbitrarily fixed distinct real numbers, unknown to the algorithm. We also allow algorithms that run in an adaptive manner, where the vantage points are selected one by one, with probes from a vantage point issued immediately upon selection, before we select the next vantage point. 

Here we aim to perform well on any specific input.
In other words, we study the \textit{instance-optimal} setting. Traditionally, instance optimality is studied by comparing the output of an algorithm to the optimal solution for the input instance. In this spirit, we assume an algorithm $OPT$ that is already aware of all the capacities, and for a specific input instance $I$, let $OPT(I)$ denote the maximum number of edges revealed by $k$ vantage points on instance $I$, i.e., obtained by enumerating all subsets of $k$ vantage points and keeping the best choice. 
In contrast, a capacity-unaware algorithm will be called \emph{$(\alpha,\beta)$-instance optimal} if after selecting at most $\alpha k$ vantage points (possibly in an adaptive manner) on the input instance $I$, it can reveal at least $OPT(I)/\beta$ many edge capacities. We can think of $\alpha\geq 1$ as the resource augmentation factor and $\beta \geq 1$ as the approximation factor for the objective function. 

We close this subsection with the following observations on the problem.
\begin{observation}
\label{lem:subpath}
Let $T_s$ be the shortest path tree rooted at a vantage point $s \in S$.
\begin{itemize}
    \item The capacities of all edges in $T_s$ that are incident to $s$ are revealed.
    \item In any root to leaf path from $s$ to $t$ on $T_s$, the capacities of the edges revealed form a record setting decreasing subsequence, where the capacity of the $i$th revealed edge is strictly smaller than all edges earlier on the path (and of course smaller than that of the $(i-1)$st revealed edge). 
    \item If $u$ is on the shortest path from $s$ to $w$, all edges that $s$ can reveal along the path from $u$ to $w$ can instead be revealed if $u$ is selected as a vantage point.
\end{itemize}
\end{observation}

\subsection{Overview of Results} 

\smallskip\noindent\textbf{Non-Adaptive Setting (with Stochastic Capacities). }
We start with the nonadaptive setting with stochastic capacity assumption, where $k$ vantage points are selected before any information on the bottleneck capacities is revealed.
Here, the ordering of the capacities $c$ is assumed to be a random permutation of $[m]$, and we look for an algorithm that maximizes the expected performance.
This problem is NP-hard, since the vertex cover problem is a special case.
We therefore aim for approximation algorithms and compare with the optimal algorithm which works under the same assumption. 

A first observation is that the expected number of revealed edges obtained using a subset $S$ of vantage points is a monotone and submodular function with respect to $S$. That is, adding a new vantage point $w$ is always beneficial and adding $w$ to a set $Y$ rather than $X$, with $X\subseteq Y$, has a diminishing return. This observation allows us to use a greedy algorithm to select the next vantage point that maximizes the expected marginal increase of the objective function, which yields a $(1-1/e)$-approximation algorithm, by known results for maximizing coverage for submodular functions \cite{Nemhauser1976-ku}. To obtain our approximation result, then, we must address the non-trivial calculation of the expected marginal increase of adding a vantage point to the set being constructed greedily. 
A crucial step here is to efficiently count the number of capacity assignments to the graph such that an edge $e$ is revealed as the result of a new vantage point. 
We show that the algorithm can be implemented in time $\Ot(kmn^3)$, with $n$ as the number of vertices, $m$ as the number of edges and $k$ as the total number of vantage points selected.

We remark that one can also resort to a sampling-based approach (e.g., see Proposition~3 in  \cite{karimi2017stochastic}), which does not provide better running time in general and can only provide an approximate solution which leads to a $1-1/e-\eps$ approximation factor.
Our deterministic algorithm provides more insight into the structure of the problem.

\smallskip\noindent\textbf{Adaptive Setting (with Worst-Case Capacities). }
In the adaptive setting, we prove the following results.
First, we show interesting lower bounds for the tradeoffs of parameters $\alpha, \beta$ and $k$ for $(\alpha, \beta)$-instance optimal algorithms, where $OPT$ selects $k$ vantage points (with the full knowledge of capacities). 
One trivial algorithm is to select all $n$ vertices as vantage points, as opposed to the $k$ vertices selected by $OPT$, obtaining $\alpha=n/k$ and $\beta =1$. We thus ask if we can lower $\alpha$ with respect to the other parameters. Our lower bounds suggest that $\alpha$ cannot be too small. 
For deterministic adaptive algorithms, we show that $\alpha^2 \beta = \Omega(n/k)$; 
for randomized adaptive ones, $\alpha \beta = \widetilde{\Omega}(\sqrt{n/k})$. These lower bounds apply for algorithms that can perform any finite computation between the selection of the $k$ vintage points. 

We then show that both bounds are tight when the underlying graph is a tree, i.e., we give vantage point selection algorithm achieving the same tradeoffs. For planar graphs, we give a deterministic algorithm that achieves a tradeoff of $\alpha\beta=O((n/k)^{2/3})$. 
Finally, while we leave open the question of existence of an $\alpha\beta=o(n)$ algorithm for general graphs, we show that if the shortest path trees are not unique, and each shortest path tree may break ties independent from others, there is a stronger lower bound of $\alpha\beta=\Omega(n^{1-\eps})$, for all $\eps>0$. 

\medskip
We present the algorithmic results in the main body of the paper. We postpone the hardness proof for both settings and some proof details to the Appendix. 

\section{$\BD$ in the Non-Adaptive Setting (with Stochastic Capacities)}
\label{sec:non-adaptive}

In this section we show a constant approximation to $\BD$ in the non-adaptive setting with stochastic capacities, which is NP-hard (\Cref{app:iohardness}).

\begin{theorem}
\label{thm:non-adaptive-main}
    There is an algorithm that gives a $(1 - 1/e)$ approximation
    to $\BD$ in time $\Ot(kmn^3)$, where $e$ is the base of the natural logarithm.
\end{theorem}

Define $R(S)$ to be the set of edges that are revealed with vertices $S$ selected, and let $f(S) = \expect{|R(S)|}$.
We prove \Cref{thm:non-adaptive-main} in two parts.
First, we show that $f$ is monotone submodular, and hence a greedy algorithm for selection into $S$ yields a $(1 - 1/e)$ approximation by the result of \cite{Nemhauser1976-ku}. This is the less interesting part and so the details are defered to \Cref{subsec:submodularity}.
We then design a value oracle for each iteration of the greedy algorithm, which is implemented by using polynomial multiplications.

\smallskip\noindent\textbf{Setup. }
Recall that $R(S)$ is the set of edges whose capacities are revealed from selecting vantage points $S$, and we would like to select vertex $w$ to maximize $f(S \cup \set{w})$. 
By linearity of expectation, $f(S \cup \set{w}) = \sum_e \prob{e \in R(S \cup \set{w})}$.
It is thus enough to show how to compute $\prob{e \in R(S \cup \set{w})}$ for an arbitrary subset $S \subseteq V$, after which we can sum over all the edges at the expense of a multiplicative $m$ factor blowup in the running time.

\smallskip\noindent\textbf{(Overview) Computing $\prob{e \in R(S)}$. }
For the remainder of this section, let us fix an edge $e = (u,v)$ for which we want to compute $p = \prob{e \in R(S)}$, and let us also have $S$ fixed.
First, we show what vertices $s \in S$ contribute to $p$.

\begin{definition}[(Non-)essential vantage points]
\label{def:ess}
    A vantage point $s \in S$ is said to be \emph{non-essential} if at least one of the following holds:
    \begin{itemize}
        \item $P(s,u)$ does not go through $(v,u)$ and $P(s,v)$ does not go through $(u,v)$.
        \item If $P(s,u)$ goes through $(v,u)$, there is some $s' \neq s$ where $s' \in P(s,u) \cap S$.
        \item If $P(s,v)$ goes through $(u,v)$, there is some $s' \neq s$ where $s' \in P(s,v) \cap S$.
    \end{itemize}
    If none of the above hold, we say that $s \in S$ is \emph{essential}.
\end{definition}

It is reasonably clear that non-essential vantage points do not have any bearing on the value of~$p$.
We now state a few observations with proofs in \Cref{app:non-adaptive}.

\begin{restatable}{observation}{noness}
\label{obs:non-ess}
        Let $S' \subseteq S$ be the set of all essential vantage points.
        Then $\prob{e \in R(S)} = \prob{e \in R(S')}$.
\end{restatable}

Now we define the minimal subgraph required to compute $\prob{e \in R(S)}$.
\begin{definition}[$T_e$, the tree rooted at $e$]
    Let $S' \subseteq S$ be the set of all essential vantage points. 
    We define $T_e = \bigcup_{s' \in S'} P(s',u) + e$ with (edge-)root $e$ (see \Cref{fig:edge-tree}).
\end{definition}

\begin{figure}[htbp]
    \centering
    \begin{subfigure}{.49\textwidth}
        \centering
        \includegraphics[width=.8\linewidth]{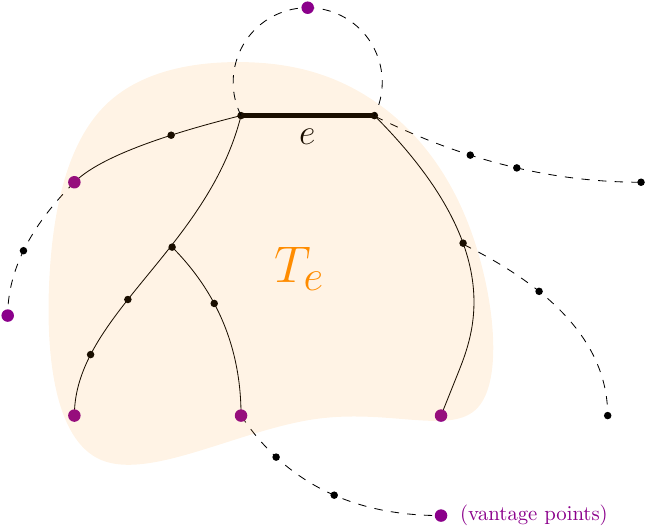}
        \caption{$T_e$, whose edges are solid. Vantage points are in magenta.}
        \label{fig:edge-tree}
    \end{subfigure}\hspace{2mm}%
    \begin{subfigure}{.49\textwidth}
        \centering
        \includegraphics[width=.8\linewidth]{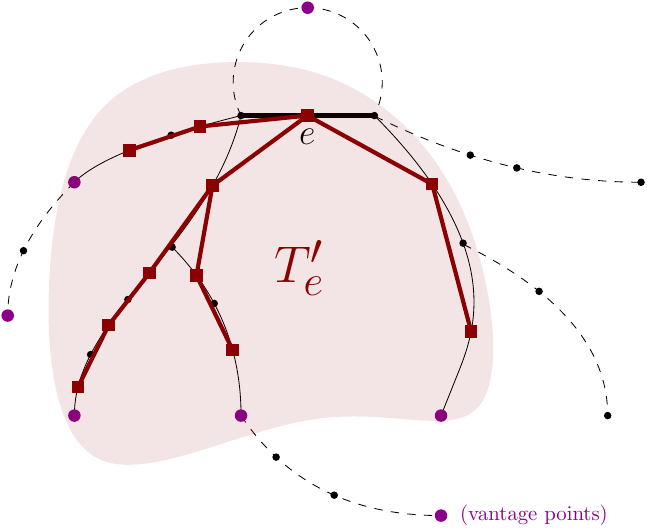}
        \caption{$T'_e$, whose vertices (squares) and edges (thick lines) are dark red.}
        \label{fig:vertex-tree}
    \end{subfigure}
    \caption{Depictions of $T_e$ and $T'_e$, the latter of which is convenient to use for $\GoodL$.}
\end{figure}

\begin{restatable}{observation}{tetreefast}
\label{obs:te-tree-fast}
    $T_e$ is a tree, with all its leaves belonging to $S$.
    Moreover, $T_e$ can be computed in $\Ot(m)$ time.
\end{restatable}

\begin{restatable}{observation}{bottontree}
\label{obs:bottontree}
    $e$ is a bottleneck edge for $S$ if and only if $c(e)$ is the smallest capacity on some (edge-)root to (vertex-)leaf path on $T_e$.
\end{restatable}

Put otherwise, using \Cref{obs:bottontree}, we have reduced the problem of computing $\prob{e \in R(S)}$ to that of counting labellings of the edges of $T_e$ with $[|V(T_e)| - 1]$ such that no two edges receive the same label, and $e$ has the smallest label on some root to leaf path of $T_e$ (recall that $T_e$ is an edge-rooted tree).
These labellings can be thought of as the order statistic of edge capacities in $T_e$.

\smallskip\noindent\textbf{Counting $T_e$ Labellings. }
For convenience, we work on the ``edge-tree'' $T'_e$ rooted at $e$
\begin{align*}
    T'_e = (E(T_e), \set{(f, \textrm{parent-edge}_{T_e}(f)) \mid f \in E(T_e)}),
\end{align*}
which allows us to think in the parlance of labelling vertices instead of edges (see \Cref{fig:vertex-tree}).
One can see that counting vertex-labellings on $T'_e$ such that there is some root to leaf path where the label of $e$ is smallest is the same as counting similarly constrained edge-labellings on $T_e$.
This leaves us with the following problem.

\begin{problem}[$\GoodL$]
    Given a rooted tree $T = (V, E)$ with root $v_1$, how many bijective labellings of $V$ with $[|V|]$ result in the label of $v_1$ being the minimum label of some root-to-leaf path?
\end{problem}

To solve $\GoodL$, however, we first consider a more constrained auxiliary problem.

\begin{problem}[$\BWL$]
    Given a rooted tree $T = (V, E)$ with root $v_1$, compute, for each $t \in [|V|]$, how many black-white colorings of $V$ there are such that:
    \begin{itemize}
        \item Exactly $t$ vertices are colored black;
        \item There is no root-to-leaf path comprising of only vertices colored white.
    \end{itemize}
\end{problem}

Conceptually, we can think of vertices colored black as those whose labels are smaller than the root's, and vertices colored white as those whose labels are no smaller than the root's.
A solution to $\BWL$ counts bad events, since, under this viewpoint we want some root-to-leaf path with all vertices colored white.

\begin{observation}
\label{obs:bwl-to-good}
    Let $B[0 \ldots t]$ be a solution to $\BWL$ on $T=(V,E)$.
    Then
    \begin{align*}
        \sum_{0 \le t \le |V|} \paren{{|V| \choose t} - B[t]}\paren{|V|-1-t}!t!
    \end{align*}
    gives a solution to $\GoodL$.
\end{observation}
\begin{proof}
    There are ${|V| \choose t} - B[t]$ colorings such that exactly $t$ vertices are colored black and there is some white root-to-leaf path.
    For each such coloring, there are exactly $(|V|-1-t)!t!$ labellings of the vertices such that white vertices have larger labels than the root's label; black vertices have smaller labels than the root's label; and there is some root-to-leaf path where the root has the smallest label.
    The result then follows from adding up good labellings where exactly $t$ vertices have a smaller label than that of the root.
\end{proof}

We now supply a polynomial time algorithm for $\BWL$, which returns a polynomial $\sum_t B[t]x^t$.

\begin{algo}
    \underline{\BWLN$(T=(V,E), v \in V)$}\+
    \\ if $v$ is leaf\+
    \\   return $x$\-
    \\ for $u \in \mathsf{children}(v)$\+
    \\   $P_u(x) \gets $\textsc{CountBlackWhiteColorings}$(T, u)$\-
    \\ $T' \gets $ subtree rooted at $v$
    \\ return $\Pi_u P_u(x) + x(1+x)^{|V(T')|-1}$
\end{algo}

\begin{claim}
\label{claim:bwl-correct}
    \BWLN returns a polynomial such that the coefficient of $x^t$ is the number of colorings of $V$ there are such that exactly $t$ vertices are colored black and there is no white root-to-leaf path.
\end{claim}
\begin{proof}
    The claim is certainly true for trees with one vertex.
    More generally, consider any vertex $v$.
    If $v$ is colored black, we can color its descendants in any way and there would be no white root-to-leaf path.
    The coefficient of $x^t$ in $x(1+x)^{|V(T')|-1}$ counts the number of ways to color $t$ vertices black, with $v$ being colored black.
    If, on the other hand, $v$ is colored white, there must be no root-to-leaf path in all of the subtrees rooted at all of its children.
    We can recursively compute the polynomials $P_u(x)$ with respect to each of $v$'s children $u$, and multiply them to get the polynomial which counts colorings in this case.
\end{proof}

\begin{claim}
\label{claim:bwl-time}
    \BWLN runs in time $\Ot(|V(T)|^2)$.
\end{claim}
\begin{proof}
    There are $|V(T)|$ polynomial multiplications involved in the total computation, where each polynomial has degree at most $|V(T)|$.
    Using \cite{harvey2021integer} to bound the running time of polynomial multiplication completes the proof.
\end{proof}

\begin{algo}
    \underline{\GoodLN$(T=(V,E), v \in V)$}\+
    \\ $P(x) \gets $ \BWLN$(T,v)$
    \\ return $\sum_{0 \le t \le |V(T)|} \paren{{|V(T)| \choose t} - [x^t]P(x)}(|V(T)|-1-t)!t!$
\end{algo}

Armed with \BWLN, we may now proceed to state complete the algorithm \GoodLN.
By \Cref{obs:bwl-to-good} and \Cref{claim:bwl-correct},  \GoodLN solves $\GoodL$ and has running time 
$\Ot(|V(T)|^2)$.

\smallskip\noindent\textbf{Back to Computing $\prob{e \in R(S)}$. }
We are finally ready to compute the desired quantity, $\prob{e \in R(S)}$: (i) Compute $T'_e$; (ii) Return $\GoodLN(T'_e, e)/\paren{\card{V(T'_e)}! (m - \card{V(T'e)})!}$.
Finishing up the proof of \Cref{thm:non-adaptive-main} is straightforward and deferred to \Cref{app:non-adaptive}.

\section{Adaptive Setting (with Worst-Case Capacities)}
\label{sec:instance}

We use notation developed in Section~\ref{sec:overview}. 
We want to select $k$ vantage points from $V$ in a given graph $G=(V,E,w)$ with unknown, but fixed, worst-case capacities $c$. Recall that an algorithm is called $(\alpha, \beta)$-instance optimal if it selects at most $\alpha k$ vantage points and reveals at least $OPT/\beta$ many bottleneck edges, where $OPT$ is the number of edges revealed (after selecting $k$ vantage points) by an algorithm that is also given all the capacities as input. Here, $\alpha \geq 1$ and $\beta\geq 1$.

\subsection{Lower bounds}\label{iolower}
In the following lower bounds, we do not limit the computation time of an algorithm -- it can be adaptive, and can do any finite computation between rounds.

\begin{restatable}[Deterministic Algorithm Tradeoff]{theorem}{thmdetiolb}\label{thmdetiolb}
For any $1\leq k\leq n$, there exists a graph $G=(V,E,w,c)$ for which any deterministic, adaptive, $(\alpha,\beta)$-instance optimal algorithm must satisfy $\alpha^2 \beta \geq \Omega(n/k).$
\end{restatable}

\begin{proof}
We will assume $k$ divides $n$ for simplicity. Let $G$ be the graph with $k$ connected components $P_{i}$, $1 \leq i \leq k$, where each $P_i$ is a path on $n/k$ vertices. Let $\mathcal{A}$ be a deterministic and adaptive algorithm. Consider the following adversarial strategy. For any $P_i$, whenever $\mathcal{A}$ selects a vantage point $v_{0}^{i}$ that is the first vantage point on $P_i$, the adversary reveals only the two neighboring edges on $v$, giving them the lowest two capacities out of the $n/k-1$ edges in $P_i$. At any point in the execution of $\mathcal{A}$, let $\{v_{0}^{i},\cdots,v_{j-1}^{i}\}$ be the vantage points selected on $P_i$ so far, and assume $\mathcal{A}$ selects a new vantage point $v_{j}^{i}$ on $P_i$. If $v_{j}^{i}$ is  not a neighbor of $v_{k}^{i}$ for any $1 \leq k \leq j-1$, the adversary reveals the two edges incident to $v_{j}^{i}$, giving them the next lowest capacities of all the edges revealed in $P_i$ so far. Otherwise, the adversary reveals either one or zero edges (depending on whether one or both neighbors of $v_{j}^{i}$ was already selected). We observe that the adversary's strategy is consistent: for any $P_i$, all edges yet to be revealed have a higher capacity than those revealed.

Let $\lambda_{i}$ be the number of vantage points selected on $P_i$ at the end of $\mathcal{A}$. We have that $\sum_{i=1}^{k} \lambda_{i} = \alpha k$, where $\alpha$ is the resource augmentation factor. The adversary now reveals the remaining edges in the following way. Fix a path $P_i$. If $\lambda_{i}=0$, the adversary gives all edges in $P_i$ decreasing capacities from the first vertex in $P_i$ (and therefore selecting that vertex would have revealed all $n/k-1$ edges). Otherwise, if $\lambda_i >0$, there is a contiguous set of $\ell_i= (n/k-1)/(\lambda_i+1) - 2$ edges such that none of the endpoints have been selected by $\mathcal{A}$ as vantage points. The adversary gives them decreasing weights, starting from, say, the first vertex in this set of edges. Selecting this first vertex would have revealed not only the $\ell_i$ edges with decreasing capacities, but also the two ``outer edges'', since their capacities are lower. Thus selecting this vertex reveals $\ell_i+2=(n/k-1)/(\lambda_i+1)$ edge capacities. We now have that $OPT \geq \sum_{i=1}^{k} (\ell_i+2)$.

Clearly $\mathcal{A}$ reveals at most $2 \alpha k$ edges, at most two per vantage point selected.  Thus 
\[\beta \geq \frac{\sum_{i=1}^{k} \frac{(n/k-1)}{\lambda_{i}+1}}{2\alpha k}, \ \text{and} \ \  \alpha k= \sum_{i=1}^{k} \lambda_i.\]
The second equation implies that there exists a set $S\subset \{1,\cdots,k\}$ such that $|S| = k/2$ and $\lambda_j \leq 2 \alpha$ for all $j \in S$. We will use $S$ to lower bound the sum in the first inequality. We have
\begin{eqnarray*}
\alpha \beta  \geq  \frac{\sum_{i=1}^{k} \frac{(n/k-1)}{\lambda_{i}+1}}{2 k}  \geq  \frac{\sum_{i \in S} \frac{(n/k-1)}{\lambda_{i}+1}}{2k}  \geq  \frac{|S|\frac{n/k -1 }{2\alpha +1 }}{2 k} 
=\frac{k/2}{2k}.\frac{n/k -1}{2 \alpha +1} \geq \frac{1}{4}.\frac{n/2k}{3\alpha},  \notag 
\end{eqnarray*}
and thus $\alpha^2 \beta \geq n/24k$, proving the theorem.
\end{proof}

\noindent\textbf{Remark:} Note that the above theorem is vacuous when $\alpha > \sqrt{n/k}$, as $\beta$ is always at least 1. We leave open the existence of a lower bound that implies $\beta = \omega(1)$ when $\alpha > \sqrt{n/k}$. As we will see, the trade off in Theorem~\ref{thmdetiolb} is tight for trees, so such a lower bound (if it exists) would involve an instance with cycles.

Next, we consider a randomized algorithm $\mathcal{A}$ that selects at most $\alpha k$ vantage points. Given an instance $G=(V,E,w,c)$, different runs of $\mathcal{A}$ may reveal different numbers of edge capacities. We define $\beta$ for a randomized algorithm to be OPT divided by the expected number of edge capacities revealed by $\mathcal{A}$ over all runs of $\mathcal{A}$ on $G$ (expectation over random coins of $\mathcal{A}$).  The proof of the following theorem can be found in \Cref{app:instance-optimal}.

\begin{restatable}[Randomized Algorithm Tradeoff]{theorem}{thmraniolb}\label{thmraniolb}
For any $1\leq k\leq n$, there exists a graph $G=(V,E,w,c)$ for which any randomized, adaptive, $(\alpha,\beta)$-instance optimal algorithm must satisfy $\alpha \beta \geq  \Omega\left(\sqrt {\frac{n}{k}} \frac{1}{(1+\log (n/k))}\right)=\widetilde{\Omega}(\sqrt{n/k})$.
\end{restatable}

\noindent\textbf{Remarks:} 
\begin{enumerate}
    \item As with Theorem~\ref{thmdetiolb}, note that Theorem~\ref{thmraniolb} is vacuous when $\alpha > \sqrt{n/k}$, as $\beta$ is always at least 1. We leave open the existence of a lower bound that implies $\beta = \omega(1)$ when $\alpha > \sqrt{n/k}$. As we will see, the tradeoff in Theorem~\ref{thmraniolb} is tight (upto the log factor) for trees, so such a lower bound (if it exists) would involve an instance with cycles.
    \item A comparison of the tradeoffs in Theorems~\ref{thmdetiolb} and ~\ref{thmraniolb} reveals that the tradeoff in Theorem~\ref{thmdetiolb} is larger than that in Theorem~\ref{thmraniolb} as long $\alpha < \sqrt{n/k}$ which by the above remark is the setting for Theorem~\ref{thmraniolb}. For example, when $\alpha=k=O(1)$, the first tradeoff gives $\beta = \Omega(n)$ whereas the second tradeoff gives $\beta = \Omega(\sqrt{n})$. This indicates that randomization may help. Indeed, it does, as we show next.
\end{enumerate}

\subsection{Upper bounds}\label{ioupper}
For general graphs, observe that there is always an algorithm that achieves the tradeoff of $\alpha \beta= O(n)$, simply because $OPT \leq k(n-1)$ (each vantage point reveals at most $n-1$ edge capacities in its shortest path tree), and the trivial algorithm that selects any set of $\alpha k$ vantage points reveals $\Omega(\alpha k)$ many capacities. Obviously, $\alpha \beta = O(n)$ is much worse than the tradeoffs in Theorems~\ref{thmdetiolb} and ~\ref{thmraniolb}. We first show that Theorems~\ref{thmdetiolb} and \ref{thmraniolb} are tight for trees, thus implying that strengthening Theorems~\ref{thmdetiolb} and \ref{thmraniolb} will require instances with cycles.

\begin{theorem}[Tight tradeoffs for trees]\label{thmtreeub}
Let $T=(V,E,w)$ be a weighted tree on $n$ vertices with fixed, unknown capacities $c$ on edges. 
\begin{enumerate}
    \item For any $1\leq k\leq n$, there exists a deterministic algorithm $\mathcal{A}_{det}$ 
    that for any given $1\leq \alpha \leq \sqrt{n/k}$, selects $\alpha k$ vantage points in $T$ such that $\mathcal{A}_{det}$ is $(\alpha,\beta)$-instance optimal 
    achieving the exact tradeoff in Theorem~\ref{thmdetiolb}, $\alpha^2 \beta = O(n/k)$.
    \item For any $1\leq k\leq n$, there exists a randomized algorithm $\mathcal{A}_{rand}$ that for any given $1 \leq \alpha \leq \sqrt{n/k}$, selects $\alpha k$ vantage points in $T$ such that $\mathcal{A}_{rand}$ is $(\alpha,\beta)$-instance optimal with $\alpha \beta = O(\sqrt{n/k})$. This tradeoff is only a factor $\log (n/k)$ worse than that in Theorem~\ref{thmraniolb}.
\end{enumerate}
\end{theorem}

\smallskip
\noindent\textbf{Proof of \Cref{thmtreeub}.} 
Both proofs will require the notion of a ``cover'' $S$, which we describe first. We take the tree $T(r)$ rooted at an arbitrary vertex $r$. Take the lowest vertex $v$ with at least $\sqrt{n}$ descendants with $v$ inclusive (i.e., any child of $v$ has fewer than $\sqrt{n}$ descendants). We put $v$ in set $S$.
We remove $v$ and its descendants and repeat until we are left with at most $\sqrt{n}$ vertices.
Clearly there are at most $\sqrt{n}$ vertices in $S$. We call this set $S$ a \emph{cover} of the tree. Sometimes we will need to vary the number of descendants: we will denote by $S(\gamma)$ the cover of size at most $n/\gamma$, generated in the same way as above starting with the lowest vertex $v$ with at least $\gamma$ descendants.
The covers satisfy the following lemma.

\begin{lemma}
\label{lem:tree-cover}
For any vertex $r'$ taken as the root of the tree, all but $\gamma$ vertices are descendants of vertices in $S(\gamma)$.
\end{lemma}
\begin{proof}
$S(\gamma)$ satisfies the requirement for the root $r$ chosen during the construction of $S(\gamma)$ by design. Now suppose we have a different root $r'\neq r$. We argue that $S(\gamma)$ again covers all but $\gamma$ vertices in the descendants of $S(\gamma)$ for $T(r')$ rooted at $r'$.
Consider the tree $T(r)$ rooted at $r$. We run a case study. 

If $r'$ is not a descendant of any vertex $v$ in $S(\gamma)$, then all descendants of $S(\gamma)$ in $T(r)$ remain to be descendants of $S(\gamma)$ in tree $T(r')$. Thus we are done. 

If $r'$ is a descendant of a vertex $v$ in $S(\gamma)$. Wlog we take $v$ as the lowest such ancestor of $r'$. We take the child $v'$ of $v$ who is the ancestor of $r'$ (or $r'$ itself). All vertices that are not in the subtree of $v'$ are now descendants of $v$ in $T(r')$. Thus they are `covered'. Now we focus on the vertices in the subtree of $v'$ in $T(r)$. Among these vertices, those that are descendants of some other vertex $w$ of $S(\gamma)$ in this subtree will continue to be descendant of $w$ in $T(r')$ -- since $r'$ is not in the subtree with root $w$. Thus the only vertices left would be those that are not in subtree of other vertices in $S(\gamma)$. There can only be at most $\gamma-1$ vertices -- otherwise $v'$ would have been selected to $S(\gamma)$ in the computation of $S(\gamma)$.
\end{proof}

\noindent\textbf{Proof of \Cref{thmtreeub} Claim 1.} The algorithm $\mathcal{A}$ selects a cover,  with any $r \in T$ as the root in the above cover-selection procedure and $\gamma = n/(\alpha k)$. Thus the size of the cover is at most $\alpha k$. If the size of the cover is strictly smaller than $\alpha k$, we select some arbitrary vantage points to make the total number of selected vantage points equal to $\alpha k$. 

Let $m$ be the total number of edge capacities revealed by $OPT$. 
Consider a vertex $r'$ chosen by the OPT solution. By Lemma~\ref{lem:tree-cover}, in the tree rooted at $r'$, all but $\gamma$ vertices of $T$ are descendants of $S(\gamma)$. Thus all the edges that are revealed by $r'$ in the OPT solution, except those that are not descendants of $S(\gamma)$, will be revealed as well by our algorithm $\mathcal{A}$. Since OPT chooses $k$ vantage points, algorithm $\mathcal{A}$ would only miss at most $k\gamma$ edges and thus reveal at least $m-k\gamma$ edges. 

Now we do a case analysis. If $m \geq 2k\gamma$, $m-k\gamma \geq m/2$. Thus $\beta\leq m/(m-k\gamma)\leq 2$. $\alpha^2 \beta \leq n/k$, since $1\leq \alpha\leq \sqrt{n/k}$.
On the other hand, if $m \leq 2k\gamma$, algorithm $\mathcal{A}$ selects $\alpha k$ vantage points, just the edges incident to these points reveal at least $\alpha k/2$ many edge capacities (in the worst case, these vertices form a matching). 
Therefore $\beta \leq \frac{m}{\alpha k/2}\leq \frac{2k\gamma}{\alpha k/2}=\frac{4n}{\alpha^2 k}.$
This immediately gives $\alpha^2 \beta \leq 4n/k$. Thus the exact tradeoff in Theorem~\ref{thmdetiolb} is achieved, and this finishes the proof of Claim 1.

\smallskip\noindent\textbf{Proof of \Cref{thmtreeub} Claim 2.} The randomized algorithm is very simple: it selects a set of $\alpha k$ vantage points uniformly at random (without repetition) from $S(\sqrt{n/k})$. The proof is similar to the deterministic case. 
Let $m$ be the number of edge capacities revealed in $OPT$. Let us consider the expected number of edges covered by the $\alpha k$ vantage points.
Note that $S(\sqrt{n/k})$ has $\sqrt{nk}$ vertices; if all selected as vantage points, would reveal at least $m - k\sqrt{n/k}=m-\sqrt{nk}$ many edges revealed by OPT. By linearity of expectation the expected number of edges that are revealed by a random choice of $\alpha k$ vantage points from $S(\sqrt{n/k})$ is at least $\frac{\alpha k}{\sqrt{nk}}(m - \sqrt{nk}).$ 
Again we consider two cases. If $m\geq 2\sqrt{nk}$, $m - \sqrt{nk}\geq m/2$. Thus
$\beta \leq 2\sqrt{nk}/(\alpha k)$, or, equivalently, $\alpha\beta \leq 2\sqrt{n/k}$.
If $m\leq 2\sqrt{nk}$, our algorithm reveals at least $\alpha k/2$ edges since we select $\alpha k$ vantage points. Thus $\beta \leq 2m/(\alpha k)\leq 4\sqrt{(n/k)}/\alpha$. Or, $\alpha\beta \leq 4\sqrt{(n/k)}$. This finishes the proof.

We remark that all algorithms are non-adaptive with polynomial run time. This is slightly surprising, since the lower bounds were on algorithms not limited by computation. We conjecture that while for trees such non-adaptive algorithms suffice, adaptivity may be needed to achieve the tradeoffs (if it is possible at all) for general graphs.

\smallskip\noindent\textbf{Planar Graphs.}
We observed in the beginning of this section that for general graphs, the trivial algorithm achieves $\alpha \beta = O(n)$, and if this is tight, a matching lower bound would involve instances with cycles. However, we show next that such an instance cannot be planar. We leave as an interesting open question whether one can design a deterministic algorithm for general graph with $\alpha \beta = o(n)$.

\begin{theorem}\label{thm:planar}
Given a weighted planar graph $G=(V, E)$ with fixed unknown capacities, and a $1\leq k \leq n$, there exists a deterministic $(\alpha, \beta)$-instance optimal algorithm with $\alpha= O((n/k)^{2/3})$ and $\beta = O(1)$, and hence $ \alpha \beta = O((n/k)^{2/3})$. 
\end{theorem}

\subsection{Stronger lower bound on general graph: Multiple shortest paths}\label{notunique}

In all of the discussion so far, we have assumed that shortest paths are unique. If there are ties and each shortest path tree may break ties independently of other trees, we can have a stronger lower bound of $\beta \geq \Omega(n^{1-\eps})$ when $\alpha=k=1$, for any $\eps>0$, between the instance optimal solution and the best solution by any (possibly randomized, adaptive) algorithm. 

We first consider a grid graph of $\sqrt{n} \times \sqrt{n}$ vertices. We wrap it around to be a torus -- the top boundary is identified as the bottom boundary and the left boundary is identified as the right boundary. The shortest path trees at all vertices, topologically, are identical -- each tree includes the horizontal edges on the row containing the root $v$, and for each vertex with the same $y$-coordinate with $v$, there are two chains of length $\sqrt{n}/2$ each. See \Cref{fig:grid} (left).

\begin{figure}[htbp]
    \begin{center}
    \includegraphics[width=0.45\linewidth]{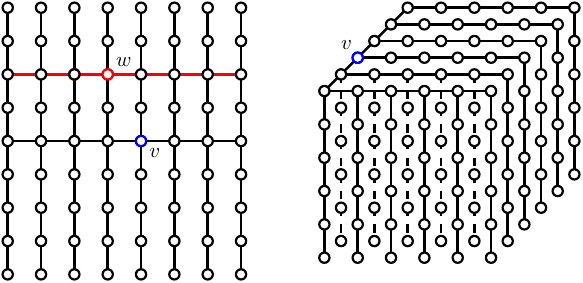} 
        \caption{(Left) The shortest path tree rooted at $v$ is shown by edges in black. The edges in red have lower capacity than edges in black. (Right)  A 3D grid graph.}
        \label{fig:grid}
    \end{center}
\end{figure}

Consider one specific vertex $v$ (in blue) and the shortest path tree rooted at $v$. We take the $n-1$ largest capacity values and assign them on the tree rooted at $v$, sorted in decreasing capacities in a round robin manner moving away from $v$. Specifically, the immediate neighbors of $v$ are given the highest capacities, those two hops away from $v$ in the tree are given the next highest capacities, and so on. This way, all edges in the tree rooted at $v$ (called $T(v)$) can be revealed if the optimal solution chooses $v$. For all other edges we assign the remaining capacity values randomly.

Now consider another vertex $w$. If $w$ does not have the same $y$-coordinate as $v$ (not on the same row as $y$), the shortest path rooted at $w$ has all edges on the same row as $w$ (shown as edges in red in \Cref{fig:grid} (left)), which have lower capacity than the edges on tree $T(v)$. Therefore, all edges other than the vertical edges on the same column as $w$ cannot be revealed. Thus, $w$ reveals at most $\sqrt{n}$ edges. If $w$ stays on the same row of $v$, $w$ may be able to reveal up to order $n$ edges but there are only $\sqrt{n}$ such vertices. Since all shortest paths look exactly the same topologically, any algorithm trying to guess the optimal tree $T(v)$ reveals in expectation $O(\sqrt{n})$ edges\footnote{Specifically, suppose an algorithm chooses a vertex $v$ with probability $p(v)$. An adversary would place the vertex on the row with the minimum total probability (which is no greater than $1/\sqrt{n}$) among all possible rows. }.

This example can be generalized to a $d$-dimensional cube of $n^{1/d}$ vertices in each dimension. See \Cref{fig:grid} (right) for an example for $d=3$. The tree rooted at a vertex $v$ takes all edges with the same coordinate of $x_1$ as $v$. For each vertex $u$ of the same $x_1$ coordinate as $v$, take the chain of vertices that share the same $x_2$ coordinate as $u$, and so on. Use the same allocation of capacities, the optimal tree can reveal $n-1$ edges while any algorithm without the knowledge of capacity distribution reveals $O(n^{1/d})$ edges. This leaves a gap of $O(n^{1-1/d})=O(n^{1-\eps})$, for any $\eps=1/d>0$, between the instance optimal solution and the best oblivious algorithm. 

\section{Open Problems}
\label{sec:open}

We conclude the paper with a couple of open problems.

\smallskip\noindent\textbf{The Adaptive Setting (with Stochastic Capacities). }
There is an intermediate model that is not addressed in this paper, which is the adaptive setting with stochastic capacities, where the ordering of capacities are a random permutation.
There are various notions of stochastic submodularity (e.g. \cite{asadpour2008stochastic,golovin2011adaptive}) for adaptive algorithms that admit good approximations via myopic algorithms.
Unfortunately, they do not seem to hold for this problem.

\smallskip\noindent\textbf{Adaptive Setting with Worst-Case Capacities.}
A notable open problem is to find algorithms that achieve the lower bounds (\Cref{thmdetiolb,thmraniolb}) for a general graph. One special parameter range is when $\beta=1$, i.e., there are $k$ vertices that, if chosen as vantage points, can reveal $OPT$ number of edges. The lower bounds for both deterministic and randomized algorithms suggest that $\alpha=\Omega(\sqrt{n/k})$. 
Now, can we find $o(n/k)$ vertices as vantage points that reveal at least $OPT$ edge capacities? Our results show that for a tree graph our algorithm with $O(\sqrt{n/k})$ vantage points suffice, and for a planar graph our algorithm with $O((n/k)^{2/3})$ vantage points suffice. We do not know any non-trivial bound of $o(n/k)$ (and $\Omega(\sqrt{n/k})$)  for a general graph nor a lower bound example that requires $\omega(\sqrt{n/k})$ vantage points. 
\section{Acknowledgments}
We thank an anonymous reviewer who caught an error in a claimed result that has now been omitted from this paper; we discuss and pose the question as the first open problem in \Cref{sec:open}.
We also thank Srinivas Narayana for discussions on capacity discovery.
\bibliography{bottleneck}

\begin{thebibliography}{10}

\bibitem{Allman1999-ex}
Mark Allman and Vern Paxson.
\newblock On estimating end-to-end network path properties.
\newblock {\em SIGCOMM Comput. Commun. Rev.}, 29(4):263--274, August 1999.

\bibitem{asadpour2008stochastic}
Arash Asadpour, Hamid Nazerzadeh, and Amin Saberi.
\newblock Stochastic submodular maximization.
\newblock In {\em Internet and Network Economics: 4th International Workshop, WINE 2008, Shanghai, China, December 17-20, 2008. Proceedings 4}, pages 477--489. Springer, 2008.

\bibitem{RIPE-Atlas}
RIPE Atlas.
\newblock User-defined measurements.
\newblock https://atlas.ripe.net/docs/getting-started/user-defined-measurements.html.

\bibitem{Augustin2007-pf}
Brice Augustin, Timur Friedman, and Renata Teixeira.
\newblock Multipath tracing with paris traceroute.
\newblock In {\em 2007 Workshop on {End-to-End} Monitoring Techniques and Services}, pages 1--8, May 2007.

\bibitem{Banerjee2000-cn}
S~Banerjee and A~K Agrawala.
\newblock Estimating available capacity of a network connection.
\newblock In {\em Proceedings {IEEE} International Conference on Networks 2000 ({ICON} 2000). Networking Trends and Challenges in the New Millennium}, pages 131--138, September 2000.

\bibitem{Beverly2016-jn}
Robert Beverly.
\newblock Yarrp'ing the {Internet}: Randomized high-speed active topology discovery.
\newblock In {\em Proceedings of the 2016 Internet Measurement Conference}, IMC '16, pages 413--420, New York, NY, USA, November 2016. Association for Computing Machinery.

\bibitem{Bottger2017-oc}
Jan B{\"o}ttger.
\newblock {\em Understanding benefits of different vantage points in today's {Internet}}.
\newblock PhD thesis, Technische Universit{\"a}t Berlin, 2017.

\bibitem{Burger2014-hb}
Valentin Burger, Matthias Hirth, Christian Schwartz, Tobias Ho{\ss}feld, and Phuoc Tran-Gia.
\newblock Increasing the coverage of vantage points in distributed active network measurements by crowdsourcing.
\newblock In {\em Measurement, Modelling, and Evaluation of Computing Systems and Dependability and Fault Tolerance}, pages 151--161. Springer International Publishing, 2014.

\bibitem{Castro2004-ll}
Rui Castro, Mark Coates, Gang Liang, Robert Nowak, and Bin Yu.
\newblock Network tomography: Recent developments.
\newblock {\em SSO Schweiz. Monatsschr. Zahnheilkd.}, 19(3):499--517, August 2004.

\bibitem{Coates2002-xq}
A~Coates, A~O Hero, III, R~Nowak, and Bin Yu.
\newblock Internet tomography.
\newblock {\em IEEE Signal Process. Mag.}, 19(3):47--65, May 2002.

\bibitem{Donnet2007-dq}
Benoit Donnet, Bradley Huffaker, Timur Friedman, and K~C Claffy.
\newblock Increasing the coverage of a cooperative internet topology discovery algorithm.
\newblock In {\em {NETWORKING} 2007. Ad Hoc and Sensor Networks, Wireless Networks, Next Generation Internet}, pages 738--748. Springer Berlin Heidelberg, 2007.

\bibitem{Donnet2005-pe}
Benoit Donnet, Philippe Raoult, Timur Friedman, and Mark Crovella.
\newblock Efficient algorithms for large-scale topology discovery.
\newblock In {\em Proceedings of the 2005 {ACM} {SIGMETRICS} international conference on Measurement and modeling of computer systems}, SIGMETRICS '05, pages 327--338, New York, NY, USA, June 2005. Association for Computing Machinery.

\bibitem{Federickson1987-pq}
Greg~N Federickson.
\newblock Fast algorithms for shortest paths in planar graphs, with applications.
\newblock {\em SIAM J. Comput.}, 16(6):1004--1022, December 1987.

\bibitem{golovin2011adaptive}
Daniel Golovin and Andreas Krause.
\newblock Adaptive submodularity: Theory and applications in active learning and stochastic optimization.
\newblock {\em Journal of Artificial Intelligence Research}, 42:427--486, 2011.

\bibitem{Goodrich1992-qh}
Michael~T Goodrich.
\newblock Planar separators and parallel polygon triangulation (preliminary version).
\newblock In {\em Proceedings of the twenty-fourth annual {ACM} symposium on Theory of Computing}, STOC '92, pages 507--516, New York, NY, USA, July 1992. Association for Computing Machinery.

\bibitem{Guerrero2010-bx}
Cesar~D Guerrero and Miguel~A Labrador.
\newblock On the applicability of available bandwidth estimation techniques and tools.
\newblock {\em Comput. Commun.}, 33(1):11--22, January 2010.

\bibitem{Harfoush2003-mw}
K~Harfoush, A~Bestavros, and J~Byers.
\newblock Measuring bottleneck bandwidth of targeted path segments.
\newblock In {\em {IEEE} {INFOCOM} 2003. Twenty-second Annual Joint Conference of the {IEEE} Computer and Communications Societies ({IEEE} Cat. {No.03CH37428})}, volume~3, pages 2079--2089 vol.3, March 2003.

\bibitem{harvey2021integer}
David Harvey and Joris Van Der~Hoeven.
\newblock Integer multiplication in time ${O}(n \log n)$.
\newblock {\em Annals of Mathematics}, 193(2):563--617, 2021.

\bibitem{Holterbach2017-ec}
Thomas Holterbach, Emile Aben, Cristel Pelsser, Randy Bush, and Laurent Vanbever.
\newblock Measurement vantage point selection using a similarity metric.
\newblock In {\em Proceedings of the Applied Networking Research Workshop}, ANRW '17, pages 1--3, New York, NY, USA, July 2017. Association for Computing Machinery.

\bibitem{Huang2020-uy}
Yuchen Huang, Michael Rabinovich, and Rami Al-Dalky.
\newblock {FlashRoute}: Efficient traceroute on a massive scale.
\newblock In {\em Proceedings of the {ACM} Internet Measurement Conference}, IMC '20, pages 443--455, New York, NY, USA, October 2020. Association for Computing Machinery.

\bibitem{Archipelago}
Y.~Hyun and K.~Claffy.
\newblock Archipelago measurement infrastructure.
\newblock http://www.caida.org/projects/ark/, 2015.

\bibitem{Jueckstock2019-ym}
Jordan Jueckstock, Shaown Sarker, Peter Snyder, Panagiotis Papadopoulos, Matteo Varvello, Benjamin Livshits, and Alexandros Kapravelos.
\newblock The blind men and the internet: {Multi-Vantage} point web measurements.
\newblock May 2019.
\newblock \href {https://arxiv.org/abs/1905.08767} {\path{arXiv:1905.08767}}.

\bibitem{karimi2017stochastic}
Mohammad Karimi, Mario Lucic, Hamed Hassani, and Andreas Krause.
\newblock Stochastic submodular maximization: The case of coverage functions.
\newblock {\em Advances in Neural Information Processing Systems}, 30, 2017.

\bibitem{Li2008-eq}
Bo~Li, Jingsha He, and Henghua Shi.
\newblock Improving the efficiency of network topology discovery.
\newblock In {\em 2008 The 3rd International Conference on Grid and Pervasive Computing - Workshops}, pages 189--194, May 2008.

\bibitem{Lichtblau2021-tr}
Franziska Lichtblau.
\newblock From the edge to the core : towards informed vantage point selection for {Internet} measurement studies, 2021.

\bibitem{Nemhauser1976-ku}
G~L Nemhauser, M~L Fisher, and L~A Wolsey.
\newblock {\em An Analysis of Approximations for Maximizing Submodular Set Functions}.
\newblock Center for Operations Research \& Econometrics, 1976.

\bibitem{OWAIS2022118134}
Mahmoud Owais.
\newblock Traffic sensor location problem: Three decades of research.
\newblock {\em Expert Systems with Applications}, 208:118134, 2022.
\newblock URL: \url{https://www.sciencedirect.com/science/article/pii/S0957417422013173}, \href {https://doi.org/10.1016/j.eswa.2022.118134} {\path{doi:10.1016/j.eswa.2022.118134}}.

\bibitem{Prasad2003-sv}
R~Prasad, C~Dovrolis, M~Murray, and K~Claffy.
\newblock Bandwidth estimation: metrics, measurement techniques, and tools.
\newblock {\em IEEE Netw.}, 17(6):27--35, November 2003.

\bibitem{Salcedo2018-kx}
Dixon Salcedo, Cesar D.~Guerrero, and Roberto Martinez.
\newblock Available bandwidth estimation tools metrics, approaches and performance.
\newblock {\em International J. Commun. Netw. Inf. Secur.}, December 2018.

\bibitem{Shavitt2011-wk}
Yuval Shavitt and Udi Weinsberg.
\newblock Quantifying the importance of vantage point distribution in internet topology mapping (extended version).
\newblock {\em IEEE Journal on Selected Areas in Communications}, 29(9):1837--1847, 2011.

\bibitem{Vardi1996-cy}
Y~Vardi.
\newblock Network tomography: Estimating {Source-Destination} traffic intensities from link data.
\newblock {\em J. Am. Stat. Assoc.}, 91(433):365--377, March 1996.

\bibitem{Vermeulen_undated-sj}
K{\'e}vin Vermeulen, Justin~P Rohrer, Robert Beverly, Olivier Fourmaux, and Timur Friedman.
\newblock Diamond-miner: Comprehensive discovery of the internet's topology diamonds.
\newblock pages 479--493.

\bibitem{yao77}
Andrew Chi-Chin Yao.
\newblock Probabilistic computations: Toward a unified measure of complexity.
\newblock In {\em 18th Annual Symposium on Foundations of Computer Science ({SFCS} 1977)}, pages 222--227, 1977.
\newblock \href {https://doi.org/10.1109/SFCS.1977.24} {\path{doi:10.1109/SFCS.1977.24}}.

\end{thebibliography}
\newpage
\appendix
\part*{Appendix}
\section{Proof of hardness}
\label{app:iohardness}

\begin{theorem}
\label{thm:bd-nph}
    The decision version of $\BD$ is $\mathsf{NP}$-hard.
\end{theorem}
\begin{proof}
    We give a reduction from $\VC$ to $\BD$.
    Let $G=(V,E), k$ be an instance of $\VC$.
    We transform it to the following instance of $\BD$: $G=(V,E,w),k$ where the weights $w$ are of the form $1 + 2^{-i}$ for edge $e_i$ so as to ensure unique shortest paths. We claim that $G$ has a vertex cover of size at most $k$ if, and only if, there exists a set of at most $k$ vantage points that reveal $|E|$ edges in expectation (and hence $|E|$ edges no matter what the capacities are).
    
    ($\implies$)
    If $G$ does indeed have a vertex cover $C$ of size $k$, then setting $S = C$ reveals every edge as a bottleneck (since edges incident to any $s \in S$ are assured to be revealed) no matter what the edge capacities are.
    
    ($\impliedby$)
    If, conversely, $G$ has no vertex cover of size $k$, then consider any $S$ of size $k$.
    Since $S$ is not a vertex cover, there is at least one edge $e$ whose endpoints are not in $S$.
    With probability at least $1/m$, $c(e)$ is the highest capacity among all the edges, and $e$ cannot be revealed unless one of its endpoints was selected. The expected number of revealed bottlenecks for $S$ is thus strictly less than $|E|$.
\end{proof}

Next, we prove hardness for the adaptive setting with worst-case capacities.

\begin{restatable}{theorem}{iohardness}
There does not exist a deterministic adaptive polynomial time algorithm $\mathcal{A}$ satisfying $\mathcal{A}(G,c)=OPT(G,c)$ for all graphs $G$ and all capacities $c$, unless $P=NP$.
\end{restatable}
\begin{proof}
    In the following, we show that if such an $\mathcal{A}$ exists, then we can solve vertex-cover in polynomial time. Assume we are given a connected graph $G$ on $n$ vertices and $m$ edges, and a number $k$, and asked if $G$ has a vertex cover of size exactly $k$. 
    
    The idea is to run $\mathcal{A}$ on $G$, starting from unknown capacities, but reveal the capacities adaptively to $\mathcal{A}$ in such a fashion that $\mathcal{A}$ is left with no choice but to solve the decision version of vertex cover. Let the current set of vantage points selected by $\mathcal{A}$ be $V_{i}=\{v_1,\cdots,v_i\}$.  Say in the next round $\mathcal{A}$ selects a vantage point $v$. For all the $d_{v,\overline{V_i}}$ edges between $v$ and a vertex not in $V_i$, we assign them the next highest $d_{v,\bar{V_i}}$ many capacities (from some arbitrarily selected alphabet) arbitrarily. This results in a consistent set of capacities: any edge capacities not yet revealed must be higher than those already revealed, and our strategy satisfies this invariant. 
    
    At the end of the execution of $\mathcal{A}$, we check if the set of vantage points $S$ selected by $\mathcal{A}$ is a vertex cover. If so, we answer YES to the decision version of vertex-cover; otherwise, we answer NO.
    
    Let $k^*$ be the size of a minimum vertex cover in $G$. If $k\geq k^*$, then $OPT(G,c)=m$ for any permutation of the capacities, in particular the capacities arrived at the end of $\mathcal{A}$'s execution, $c_{end}$. Therefore, in this case, $\mathcal{A}(G,c_{end})=m$ also, and that the set $S$ of vantage points selected must be a vertex cover, for if an edge $e$ is not covered by $S$, then its capacity is larger than all revealed edges, and could not have been revealed to $A$.
    
    On the other hand if $k<k^*$, then the set of vantage points selected by $\mathcal{A}$ cannot be a vertex-cover by definition of $k^*$, and so we answer correctly. Hence, the theorem is proved.
\end{proof}
\section{Omitted Proofs: $\BD$ in the Stochastic, Non-Adaptive Setting}
\label{app:non-adaptive}

\subsection{Submodularity of $f$}
\label{subsec:submodularity}
\begin{lemma}
\label{lem:nonad-submod}
$f(X)$ is monotone and submodular as a function of $X$:
\begin{enumerate}
    \item Monotonicity: $f(X\cup \set{w})\geq f(X)$;
    \item Submodularity: If $X\subseteq Y$, then $f(X\cup \set{w})-f(X)\geq f(Y\cup \set{w})-f(Y)$.
\end{enumerate}
\end{lemma}
\begin{proof}
    We prove each claim separately.
    \begin{enumerate}
        \item Let $X_e$ (resp. $X^w_e$) be the indicator for $e$ being revealed from selecting $X$ (resp. $X \cup \set{w}$).
        Then, by linearity, $f(X)=\sum_e \E(X_e)=\sum_e \Pr[X_e=1]$.
        It is clear that $\Pr[X^w_e=1] \ge \Pr[X_e=1]$, and so $f$ is monotone.
        \item Let $X_e$ (resp. $X^w_e, Y_e, Y^w_e, (Y \setminus X)_e, w_e$) be the indicator for $e$ being revealed from selecting $X$ (resp. $X \cup \set{w}, Y, Y \cup \set{w}, Y \setminus X, \set{w}$).
        Then
        \begin{align*}
            f(X\cup \set{w}) - f(X) &= \sum_e \prob{X^w_e=1} - \prob{X_e=1}
            \tag{Linearity of expectation}
            \\
            &= \sum_e \prob{X_e=0 \wedge w_e=1}
            \tag{Law of total probability}
            \\
            &= \sum_e \prob{X_e=0 \mid w_e=1}\prob{w_e=1}
            \\
            &\ge \sum_e \prob{X_e=0 \wedge (Y \setminus X)_e=0 \mid w_e=1}\prob{w_e=1}
            \\
            &= \sum_e \prob{Y_e=0 \mid w_e=1}\prob{w_e=1}
            = f(Y\cup \set{w}) - f(Y),
        \end{align*}
        and so $f$ is submodular.
    \end{enumerate}
\end{proof}

By the results on maximizing submodular functions~\cite{Nemhauser1976-ku}, incrementally building up a set $S$ by adding the vantage point maximizing $f(S \cup \set{w})$ until $|S|=k$ yields a $(1-1/e)$ approximation algorithm.

\noness*
\begin{proof}
    In short, this follows from \Cref{lem:subpath}.
    
    Let $s' \in S'$.
    
    Querying $s'$ can only reveal bottleneck edges in $T_{s'}$, the shortest path tree rooted at $s'$.
    Consequently, if $P(s',u)$ does not go through $(v,u)$ and $P(s',v)$ does not go through $(u,v)$, then $\prob{e \in R(S)} = \prob{e \in R(S - s')}$.

    On the other hand, note that at most one of $P(s',u)$ goes through $(v,u)$ and $P(s'v)$ goes through $(u,v)$.
    Without loss of generality, say $P(s',u)$ goes through $(v,u)$.
    Since $s' \in S'$, there must be a vertex $s \in S$ where $s \neq s'$ such that $s \in P(s',u)$.
    The event that querying $s'$ reveals $(v,u)$ as a bottleneck is a subset of the event that querying $s$ reveals $(v,u)$ hence $\prob{e \in R(S)} = \prob{e \in R(S - s')}$.

    Finally, observe that a vertex being non-essential only depends on $G$ and $e$, and since $s'$ was chosen arbitrarily, $\prob{e \in R(S)} = \prob{e \in R(S \setminus S')}$.
\end{proof}

\tetreefast*
\begin{proof}
    $T_e$ being a tree follows immediately since, in our model, shortest paths are unique.
    That all its leaves are in $S$ follows from \Cref{def:ess}, the definition of essential vantage points.

    To compute $T_e$, we need only find $T_u$ (resp. $T_v$), the shortest path tree rooted at $u$ (resp. $v$), which we can compute in $\Ot(m)$ time using Dijkstra's algorithm.
    We can then find $T_e$ by running a breadth first search in $T_u$ and $T_v$, truncating at
    \begin{itemize}
        \item Subtrees which do not have $(u,v)$ as an ancestor;
        \item Vertices $s \in S$;
    \end{itemize}
    and finally iteratively removing leaves that are not in $S$.
    Since $P(s,u) = P(u,s)$, this produces $\bigcup_{s' \in S'} P(s',u) + e$.
    The running time follows.
\end{proof}

\bottontree*
\begin{proof}
    This follows, in short, from \Cref{lem:subpath}.
    \\
    $(\implies)$
    \\
    Suppose $e$ is a bottleneck edge for $S$.
    Let $s \in S$ be the vantage point that is closest to $e$, for which $e$ is a bottleneck edge on, say, $P(s,v)$.
    Since $s$ is the closest vantage point, it is essential and so $P(s,v) \subseteq T_e$.
    In particular, $s$ is a leaf with $c(e)$ being the smallest capacity on $P(v,s)$.
    \\
    $(\impliedby)$
    \\
    The converse is proved similarly.
    Suppose for a leaf $s \in S$, the capacity $c(e)$ is smallest in say $P_{T_e}(s,v)$.
    Then $e$ is a bottleneck edge of $P(s,v) = P_{T_e}(s,v)$.
\end{proof}

\begin{proof}[Proof of \Cref{thm:non-adaptive-main}]
    The correctness and approximation factor of \textsc{GreedyNonAdaptive} have been established earlier. 
    For running time,
    \textsc{GreedyNonAdaptive} takes $k$ rounds, and computing each step takes $nm$ inner-iterations, with each iteration computing $\prob{e \in R(S)}$ for some $e \in E$ and $S \subseteq S'$.
    Each such computation takes $\Ot(m)$ time to compute $T'_e$ (\Cref{obs:te-tree-fast}) and $\Ot(n^2)$ time to run \GoodLN. 
    \end{proof}

\section{Omitted Proofs: Adaptive setting with worst-case capacities}
\label{app:instance-optimal}

\subsection{Lower Bound Proofs in \Cref{iolower}}

\begin{proof}[Proof of \Cref{thmraniolb}]
We will first invoke Yao's principle \cite{yao77}, and try to lower bound the performance of a deterministic algorithm on the following distribution of inputs. In fact, our distribution inputs will all have the same underlying graph, but different set of capacities.

We will assume $k$ divides $\sqrt{n}$ for simplicity. Let $G$ be the disconnected graph with $\sqrt{kn}$ components $P_{i}$, $1 \leq i \leq \sqrt{kn}$, where each $P_i$ is a path on $\sqrt{n/k}$ vertices. Consider assigning capacities in the following randomized procedure. First, select a random set $S \subset \{1,\cdots,\sqrt{kn}\}$, $|S|=k$. For any $i \in S$, assign the edges in $P_i$ decreasing capacities, starting from the leftmost vertex. We will call these the ``good'' paths. For any $i \notin S$, assign the edges in $P_i$ a random permutation of integers in the range $[(i-1)\sqrt{n/k}, i\sqrt{n/k})]$. We will call these the ``bad'' paths.
The optimal adaptive algorithm will select the left endpoints of all $k$ good paths and thus reveal roughly $k \cdot \sqrt{n/k}= \sqrt{kn}$ edges.

Our hard distribution $\mathcal{D}$ will be the uniform distribution over all capacities obtained following the above procedure on $G$.
We now upper bound the expected number of edges revealed by any deterministic algorithm $\mathcal{A}$ over the distribution of inputs $\mathcal{D}$. Recall that $\mathcal{A}$ is allowed to select $\alpha k$ many vantage points. First, observe that if $\alpha \geq \sqrt{n/k}$, the theorem is trivially true since $\beta \geq 1$. So we will assume $\alpha < \sqrt{n/k}$, which means that $\alpha k < \sqrt{kn}$, the total number of paths. Thus, $\mathcal{A}$ does not have enough vantage points to cover all paths. 

We can assume that $\mathcal{A}$ does not select two vantage points on the same bad path, as this will only boost the chance of $\mathcal{A}$ hitting the good paths.  To argue this assumption, first, we boost $\mathcal{A}$ by giving it for free all edge capacities if it selects any (not necessarily the first) vantage point on a good path, so there is never a need to select another vantage point on a good path. If $\mathcal{A}$ selects a vantage point on a bad path, in expectation it reveals $O(\log (\sqrt{n/k}))=O(\log (n/k))$ many edge capacities. This is because on a random permutation of length $n$ there is a decreasing sequence of $O(\log n)$ record-setting edge capacities, in expectation. This follows from Lemma~\ref{lem:subpath}: the probability that the $i$th edge has a capacity lower than the previous $i-1$ edges is at most $1/(i-1)$, and the expected number of edges revealed can be shown to be upper bounded by $2(1+1/2+\cdots+1/(n/2)) = O(\log n)$. $\mathcal{A}$ does not have enough vantage points to cover all paths, and therefore it is always optimal to not place another vantage point on a bad path, and instead place it on an uncovered path.

We have that the expected number of good paths revealed by $\mathcal{A}$ is $\alpha k( k / \sqrt{kn})= \alpha k \sqrt{k/n}$, and thus the expected number of good edges revealed is $( \alpha k \sqrt{k/n})(\sqrt{n/k}) = \alpha k$. 

Now observe that the expected number of bad edges revealed is $O(\alpha k \log (n/k))$, at most $O(\log (n/k))$ per vantage point on bad path. Thus the expected number of total edges revealed is still $O(\alpha k \log (n/k))$, whereas $OPT \geq \sqrt{kn}$. Thus $\beta = \Omega(\sqrt{kn}/(\alpha k \log (n/k)))$, proving the theorem.
\end{proof}

\subsection{Upper Bound Proofs in \Cref{ioupper}}
\begin{proof}[Proof of \Cref{thm:planar}]
We use the $r$-division of a planar graph~\cite{Goodrich1992-qh, Federickson1987-pq}. Specifically, an $r$-division is a subdivision of $G$ into $O(n/r)$ pieces satisfying three conditions:
Each piece has $O(r)$ vertices;
Each piece has $O(\sqrt{r})$ boundary vertices (i.e., vertices shared with other pieces);
Each piece has $O(1)$ holes (faces of the piece that are not faces of $G$.
Such a $r$-division can be computed in $O(n)$ time~\cite{Goodrich1992-qh}. 

Suppose the optimal solution chooses $k$ vertices $S$, our algorithm basically chooses all vertices on the boundary of all pieces in an $r$-division, for a value $r$ to be decided later. There are $O(n/r) \cdot O(\sqrt{r})=O(n/\sqrt{r})$ boundary vertices. $\alpha=O(n/(k\sqrt{r}))$.

For any vertex $v$ in a piece $R$ without any vantage point from the optimal solution $S$, the shortest path from any vantage point $u\in S$ (which is outside the piece) to $v$ has to go through one of the boundary vertices of the piece $R$, say $w$. Thus, all edges that are revealed along the path $P(w,v)$, which is a subpath of $P(u, v)$, are also revealed by our algorithm (by $w$ in particular). Therefore the only edges our algorithm may have missed would be the edges of the shortest path tree rooted at $u$ that stay inside $R$, for each $u\in S$. There are at most $O(r)$ such edges for each vantage point $u\in S$ and a total of $O(rk)$ such edges for all vertices in $S$. 

Suppose that the optimal solution reveals $m$ edges. We will reveal at least $m-crk$ edges, for some constant $c$. If $m \geq 2crk$, $m-crk\geq m/2$. Thus $\beta\leq 2$ and $\alpha \beta \leq O(n/(k\sqrt{r}))$. 
If $m\leq 2crk$, we reveal at least $\Omega(\alpha k)$ edges since we select $\alpha k$ vertices. Thus $\beta \leq O((2crk)/(\alpha k))$. Or, $\alpha \beta \leq O(r)$. Now we take $r=(n/k)^{2/3}$ to balance the two terms. This gives an upper bound for $\alpha \beta =O((n/k)^{2/3})$.
\end{proof}

\end{document}